%% file: arxiv.tex
\documentclass{IEEEtran4PSCC}
\ifCLASSINFOpdf
   \usepackage[pdftex]{graphicx}
\else
   \usepackage[dvips]{graphicx}
\fi

\usepackage[cmex10]{amsmath}
\usepackage{amsfonts,amssymb,amsthm}
\usepackage[mathscr]{euscript}

\usepackage{algorithm}
\usepackage{algpseudocode}
\usepackage[]{threeparttable}
\usepackage{booktabs}

\usepackage{paralist}
\allowdisplaybreaks

\usepackage[disable]{todonotes}
\newcommand{\mt}[1]{\todo[inline,color=blue!20]{MT: #1}}

\newtheorem{theorem}{Theorem}
\theoremstyle{remark}

\usepackage{mathtools}

\input{definitions}

\begin{document}

\author{
\IEEEauthorblockN{Beno\^{i}t Jeanson\IEEEauthorrefmark{1}\IEEEauthorrefmark{2},  Mathieu Tanneau\IEEEauthorrefmark{3}, Simon~H.~Tindemans\IEEEauthorrefmark{2}}
\IEEEauthorblockA{\IEEEauthorrefmark{1}Dept. of Electrical Sustainable Energy, Delft University of Technology, Delft, The Netherlands 
\\
\IEEEauthorrefmark{2}\textit{CRESYM}, Brussels, Belgium
\\
\IEEEauthorrefmark{3}
Georgia Institute of Technology, Atlanta, GA, United States
}}

\title{Scalable Iterative Algorithm for Solving Optimal Transmission Switching with De-energization
}

\maketitle

\input{abstract}
\input{introduction}

\input{formulation}
\input{completeHeuristic}

\input{results}

\input{conclusion}
\section*{Acknowledgements}
Benoît Jeanson and Simon Tindemans were supported by the CRESYM OptGrid project funded by RTE – Réseau de Transport d’Électricité, and Mathieu Tanneau by Los Alamos National Laboratory under Award No. C4995.
\bibliographystyle{IEEEtran}

\bibliography{refs}
\end{document}

%% file: definitions.tex
\newcommand{\pfmax}{\bar{f}}
\newcommand{\pgref}{\hat{g}}  
\newcommand{\pdref}{\hat{d}}  

\newcommand{\va}{\boldsymbol{\theta}}
\newcommand{\pg}{\mathbf{g}}  
\newcommand{\pd}{\mathbf{d}}  
\newcommand{\pf}{\mathbf{f}}  
\newcommand{\cf}{\mathbf{f}^\star}  

\newcommand{\vbranch}{\mathbf{v}}  
\newcommand{\vfsol}{\mathbf{\Tilde{v}}} 
\newcommand{\vsol}{\mathbf{\bar{v}}} 
\newcommand{\LL}{\text{LL}}  
\newcommand{\ol}{{ol}} 
\newcommand{\vpi}{\boldsymbol{\pi}}

\newcommand{\org}{o}
\newcommand{\dst}{d}

\newcommand{\edges}{\mathcal{E}}
\newcommand{\buses}{\mathcal{V}}
\newcommand{\contingencies}{\mathcal{C}}
\newcommand{\allcases}{\bar{\mathcal{C}}} 
\newcommand{\contingenciesSa}[1]{\textbf{C}_{#1}} 
\newcommand{\crv}{\mathbf{C}^{RV}} 

\newcommand{\GC}{\mathbf{K}} 
\newcommand{\gc}{\mathcal{K}} 

\newcommand{\PBO}{SBS} 
\newcommand{\FS}{FS} 
\newcommand{\MB}{MB} 

\newcommand{\nbhop}{n_{hops}}
\newcommand{\nbhopO}{nh_0}
\newcommand{\nbhopmax}{nh_{max}}

\newcommand{\SA}{\textbf{SecurityAnalysis}}
\newcommand{\MostConstraining}{\textbf{MostConstraining}}
\newcommand{\ReduceViol}{\textbf{ReduceViolations}}
\newcommand{\RemoveUO}{\textbf{RemoveUnnecessaryOpenings}}
\newcommand{\hop}{\textbf{Hop}}
\newcommand{\VBf}[2]{VB_{#1}^{#2}}

\newcommand{\originAlg}{JT25}
\newcommand{\completeAlg}{Reformulated}

%% file: abstract.tex
\begin{abstract}
    Transmission System Operators routinely use transmission switching as a tool to manage congestion and ensure system security.
    Motivated by sub-transmission operations at RTE, this paper considers the Optimal Transmission Switching with De-energization (OTSD), which captures potential loss of connectivity (and therefore localized blackout) following loss of transmission elements.
    While directly relevant to real-life operations, this problem has received very little attention in the literature.
    The paper proposes a new mixed-integer linear programming formulation for OTSD that represents post-contingency loss of connectivity without requiring additional binary variables.
    This new formulation provides the foundation for a fast, iterative heuristic algorithm.
    Computational experiments confirms that state-of-the-art optimization solvers struggle to solve the extensive formulation of OTSD, often failing to find even trivial solutions within reasonable time.
    In contrast, numerical results demonstrate the efficiency of the proposed heuristic, which finds high-quality feasible solutions 100-1000x faster than using Gurobi.
    
\end{abstract}

%% file: introduction.tex
\section{Introduction}
\label{sec:intro}

    Transmission switching is routinely used by Transmission System Operators (TSOs) to manage grid congestion, reduce operational costs, and ensure voltage stability and {`N-1'} security \cite{Hedman2011_ReviewOTS}.
    This operational tool is particularly valuable under high renewable penetration \cite{littleOptimalTransmissionTopology2021}.
    Indeed, unlike redispatching and load shedding, topology actions -- such as branch opening and bus splitting -- are generally no/low-cost actions.
    Nevertheless, identifying the best switching pattern, referred to as the \emph{Optimal Transmission Switching} (OTS) problem, is highly combinatorial and generally intractable to solve exactly.
    Therefore, system operators often spend a considerable amount of time manually designing good operational switching patterns, which is both time-consuming and potentially sub-optimal.

    This paper focuses on the OTS problem arising in RTE -- the French TSO -- within its subtransmission network, operating at 63 kV and 90 kV.
    In these areas, power lines may be located above roads or buildings, hence, the consequences of exceeding thermal limits on lines are considered too risky for assets and/or persons. 
    This results in the following operational rule: line overloads are forbidden both in the base (`N') case and after a contingency (`N-1').
    As a consequence, in the event of a contingency, it may be preferable to de-energize  part of the grid, resulting in a local blackout, rather than violate thermal limits and risk damage to life and/or assets.
    Nevertheless, despite its relevance for real-life operations, this setting has received little attention in the literature.

    \subsection{Related Works}
    \label{sec:intro:related_works}

        The following subsection reviews prior works addressing OTS and its security-constrained variants.

        Following the seminal work of \cite{bacherNetworkTopologyOptimization1986}, the OTS problem has received growing attention, with several formulations and algorithmic approaches proposed over the years \cite{Hedman2011_ReviewOTS}.
        These techniques range from metaheuristics to mixed-integer programming (MIP) approaches and, more recently, reinforcement learning \cite{sarOptimizingPowerGrid2025}.
        \mt{Done re-writing until here, below is copy-pasted}
        \mt{@BJ: update this section by including additional references}
        However, adoption of OTS by TSOs has been slow. 
        This may stem from the counterintuitive idea that fewer connected branches can alleviate congestion, exemplifying Braess’ paradox. 
        This phenomenon is further examined in \cite{floresAlternativeMathematicalModels2021a}.
        
        The OTS problem is rarely addressed alone, and is often formulated combined with economic dispatch (ED) and/or unit commitment (UC).
        Indeed, the first MIP approach to the OTS problem already included a basic ED \cite{fisherOptimalTransmissionSwitching2008}.
        It is the foundation of many further works:  the authors of \cite{hedmanOptimalTransmissionSwitching2009} developed a \emph{security-constrained} OTS with ED, soon complemented by generation UC \cite{hedmanCoOptimizationGenerationUnit2010}.
        The algorithm in \cite{hanOptimalTransmissionSwitching2023} takes up the challenge of combining in a three-stage problem, the branch-only opening OTS problem with a linearized AC power flow, the `N-1' rule and uncertainties on RES using a column-and-constraint generation algorithm and Dantzig-Wolfe approaches.
        The recent review on OTS optimization \cite{numanRoleOptimalTransmission2023} focuses solely on branch-switching methods, yet shows that even this simplified form remains a problem of active research and  is increasingly discussed alongside other developments pursued by TSOs, such as Dynamic Thermal Rating and Energy Storage.

        Even in its most basic form, which consists of assessing the grid situation without contingency and under the DC-approximation that drastically simplifies the problem, the problem is already NP-hard as demonstrated in \cite{lehmannComplexityDCSwitchingProblems2014}. 
        The \emph{security-constrained} OTS problem is even more challenging, but the additional complexity is usually considered unavoidable, as the `N-1' rule is a key policy for operations. 
        This significantly increases the problem size, because physical variables -- typically the flows on the branches -- must be assessed for the base case as well as for each post-contingency case.
        This scaling effectively prohibits the use of AC power flow equations inside the problem formulation, so that practical implementations rely on the DC power flow approximation. 
        This has significant limitations on aspects such as voltage, short-circuits, and stability, but is a common first simplification step with the perspective of scaling up.
        Indeed, \cite{liSolvingOptimalTransmission2023a} confirmed the validity of this approach specifically for the OTS problem, when limiting the number of branch openings and picking up many sub-optimal solutions to be validated with an AC model.

        To handle the size increase induced by the `N-1' cases, \cite{heidarifarOptimalTransmissionReconfiguration2014} applies a Benders decomposition to solve the security-constrained OTS problem with OPF combining line switching and bus splitting. However, the experiments are limited to the IEEE 14-bus system.

        This paper builds upon the work in \cite{jeanson2025riskbasedapproachoptimaltransmission}, which introduced a de-energizing operational policy for the OTS MIP model of \cite{fisherOptimalTransmissionSwitching2008}, whose results are also limited to the IEEE 14-bus system.
        The formulation becomes rapidly intractable for larger networks due to the computational cost of assessing the connectivity after tripping to identify the energization state of each bus.
        As such, the problem is far more complex than the usual security-constrained OTS problem as commonly addressed, especially because the on-off behavior applies at two stages: the state of branches and its consequences on the energization state of the buses and the corresponding load and generation injections.
        
        We point out that \cite{dingRobustOptimalTransmission2016} addresses a similar level of complexity and also develops a two-stage algorithm to solve the branch-only opening security-constrained OTS problem with corrective remedial actions.
        It is similar in that injections are adjusted in response to a tripping, but simpler to solve since the `N-1' cases remain connected and no de-energization is to be considered.

        In practice in RTE, sub-transmission systems can be split in areas that rarely exceed 100 buses.
        Thus, an algorithm that solves the OTS problem with De-energization (OTSD) at such local scales is a key component with the perspective to be integrated in a wider optimization program tackling the whole grid such as described in \cite{khanabadiDecentralizedTransmissionLine2018} and \cite{henkaPowerGridSegmentation2022}.

        \mt{end of copy-paste}

        \subsection{Contributions and Outline}
        \label{sec:intro:contributions}

        \mt{This (sub)section was re-written}
    
        Motivated by the operational practices of RTE for sub-transmission systems, this paper improves the scalability of OTSD.
        These advances mark a first step towards solving OTSD exactly on sub-transmission networks with a few hundred buses.
        Its contributions are as follows:
        \begin{enumerate}
            \item It proposes a new MILP formulation for OTSD that does not require binary variables for representing post-contingency loss of connectivity (see Theorems \ref{thm:pi} and \ref{thm:pi_continu}), a significant improvement compared to \cite{jeanson2025riskbasedapproachoptimaltransmission}.
            \item It introduces a fast algorithm for finding feasible solutions by integrating a column-and-constraint algorithm and variable neighborhood local search strategies.
            \item It reports computational experiments on networks with up to 200 buses.
        \end{enumerate}
        
        The rest of the paper is structured as follows.
        Section \ref{sec:problem} presents a new formulation for OTSD.
        Section \ref{sec:heuristic} describes the proposed heuristic algorithm.
        Numerical results are presented in Section \ref{sec:results} and Section \ref{sec:conclusion} concludes the paper.

%% file: formulation.tex
\section{Problem Formulation}
\label{sec:problem}

This section presents the OTSD problem as an MILP.
The proposed formulation builds on previous work \cite{jeanson2025riskbasedapproachoptimaltransmission}, but uses fewer variables to represent N-1 de-energization.

\subsection{Definitions}
\label{sec:formulation:notations}
    
    Consider a power grid represented as a graph $G = (\buses, \edges)$, where $\buses$ is the set of buses, and $\edges$ is the set of branches (i.e. lines and transformers).
    The generation and load at bus $i \, {\in} \, \buses$ are denoted by $\pgref_{i}$ and $\pdref_{i}$, respectively.
    The susceptance and thermal limit of branch $e \, {\in} \, \edges$ are denoted by $b_{e}$ and $\pfmax_{e}$.
    Note that, in this work, these are not decision variables.
    
    The \emph{origin} and \emph{destination} nodes of branch $e \, {\in} \edges$ are denoted by $\org(e)$ and $\dst(e)$, respectively, defining the sign of flows on branch $e$.
    The set of branches leaving (resp. entering) bus $i$ is denoted by $\edges^{+}_{i}$ (resp. $\edges^{-}_{i}$).
    Given $\buses' \subseteq \buses$, the \emph{cutset} $\gc = [\buses', \buses {\setminus} \buses'] \, {\subseteq} \, \edges$ is the set of branches with exactly one end point in $\buses'$.
    The set of all cutsets that separate bus $i$ from a predefined ``reference" bus $r \, {\in} \, \buses$ is denoted by $\GC_{i} = \{ [\buses', \buses {\setminus} \buses'] \, | \, i \, {\in} \, \buses', r \notin \buses' \}$.
    In this paper, the reference bus is chosen as a bus known to remain energized under any situation.
   
    A \emph{contingency} refers to the simultaneous loss of a set of branches $c \, {\subseteq} \, \edges$.
    The set of energized and de-energized buses under contingency $c$ are denoted by $\buses^{c}$ and $\bar{\buses}^{c}$, respectively.
    Therein, a bus is \emph{energized} if it is still connected to the grid's main component, and \emph{de-energized} otherwise.
    For simplicity, the paper defines the ``main component" as the set of buses that remain connected to the reference bus $r$.
    
    Finally, let $\contingencies \, {\subseteq} \, 2^{\edges}$ denote a set of monitored contingencies, and define $\allcases \, {=} \, \{\emptyset\} \cup \contingencies$ where $c \, {=} \, \emptyset$ denotes the base case (no contingency).
    The probability of occurrence of contingency $c$ is denoted by $p_{c}$.
    While the proposed methodology supports an arbitrary set of monitored contingencies, for ease of presentation, the rest of the paper considers the ``N-1" criterion, i.e., $\contingencies = \{ \{e\} \}_{e \, {\in} \, \edges}$.

\subsection{Base Case Variables and Constraints}
\label{sec:formulation:basecase}

    The paper considers an OTSD formulation where active power generation in the base case is \emph{not} a decision variable.
    This reflects operating practices are RTE, where generation re-dispatch is avoided whenever possible.
    The phase angle at bus $i \, {\in} \, \buses$ in the base case is denoted by $\va_{i}^{\emptyset}$.
    The open/closed status of branch $e \, {\in} \, \edges$ is denoted by $\vbranch_{e}^{\emptyset}$, where
    \begin{align}
        \label{eq:basecase:vbranch}
        \vbranch^{\emptyset}_{e} \, {\in} \, \{0, 1\}
            && \forall e \, {\in} \, \edges,
    \end{align}
    equals $0$ if branch $e$ is open (i.e., disconnected) and $1$ otherwise.
    The power flow on branch $e$ is denoted by $\pf_{e}^{\emptyset}$.
    DC power flow equations are expressed using Ohm's law and Kirchhoff's current law as
    \begin{align}
        \label{eq:basecase:ohm}
         \pf_{e} &= b_{e} \vbranch^{\emptyset}_{e} (\va^{\emptyset}_{\dst(e)} - \va^{\emptyset}_{\org(e)}),
            && \forall e \, {\in} \, \edges,\\
        \label{eq:basecase:kcl}
        \pgref_{i} + \sum_{e \, {\in} \, \edges^{-}_{i}} \pf^{\emptyset}_{e}  &= \pdref_{i} + \sum_{e \, {\in} \, \edges^{+}_{i}} \pf^{\emptyset}_{e},
        && \forall i \, {\in} \, \buses,
    \end{align}
    and thermal limits are enforced as
    \begin{align}
        \label{eq:basecase:thermal}
        |\pf_{e}| & \leq \pfmax_{e},
            && \forall e \, {\in} \, \edges.
    \end{align}

    Similarly to \cite{jeanson2025riskbasedapproachoptimaltransmission}, the paper enforces connectedness in the base case using virtual flows.
    Thereby, let $\delta \, {\in} \, \mathbb{R}^{|\buses|}$ denote the vector of virtual injections, where $\delta_{r} \, {=} \, 1{-}|\buses|$ and $\delta_{i \neq r} \, {=} \, 1$, and let $\cf_{e}$ denote the virtual flow on branch $e \, {\in} \, \edges$.
    Connectedness is then enforced as follows
    \begin{align}
        \label{eq:basecase:shadow_flow}
        \sum_{e \, {\in} \, \edges^{-}_{i}} \cf_{e}  &= \delta_{i} + \sum_{e \, {\in} \, \edges^{+}_{i}} \cf_{e},
        && \forall i \, {\in} \, \buses,\\
        \label{eq:basecase:shadow_bigM}
        |\cf_{e}| & \leq \vbranch^{\emptyset}_{e} |\buses|,
            && \forall e \, {\in} \, \edges.
    \end{align}
    It is easy to see that \eqref{eq:basecase:shadow_flow}-\eqref{eq:basecase:shadow_bigM} is feasible if and only if branch openings $\vbranch^{\emptyset} \, {\in} \, \{0, 1\}^{E}$ are such that the network is connected.

\subsection{Post-contingency Variables and Constraints}
\label{sec:formulation:contingency}

    The post-contingency constraints capture 1) the potential loss of connectivity (de-energization), and 2) post-contingency thermal limits on branches.
    
    The post-contingency energized/de-energized status of bus $i \, {\in} \, \buses$ is denoted by $\vpi^{c}_{i} \, {\in} \, \{0, 1\}$, where $\vpi^{c}_{i} = 1$ (resp. $0$) if bus $i$ is energized (resp. de-energized) under contingency $c \, {\in} \, \contingencies$.
    The post-contingency generation and demand at bus $i$ are denoted by $\pg^{c}_{i}$ and $\pd^{c}_{i}$, respectively, and computed as
    \begin{align}
        \label{eq:contingency:demand}
        \pd_{i}^{c} &= \vpi_{i}^{c} \pdref_{i}
            && \forall i, c \, {\in} \, \buses \times \contingencies,\\
        \label{eq:contingency:generation}
        \pg_{i}^{c} &= \boldsymbol{\sigma}^{c} \vpi_{i}^{c} \pgref_{i}
            && \forall i, c \, {\in} \, \buses \times \contingencies,
    \end{align}
    where variable $\boldsymbol{\sigma}^{c}$ enforces a post-contingency rule that all (energized) generators' dispatches are adjusted proportionally to balance the system.
    Note that constraints \eqref{eq:contingency:demand}, \eqref{eq:contingency:generation} imply $\pd^{c}_{i} = \pg^{c}_{i} = 0$ if bus $i$ is de-energized.
    The loss of load under contingency $c$ is the sum of loads located at de-energized buses
    \begin{align}
        \label{eq:contingency:load_load}
        \LL^{c} &= \sum_{i \, {\in} \, \buses} \pdref_{i} - \pd_{i}^{c}
            && \forall c \, {\in} \, \contingencies.
    \end{align}
    
    The open/close status and power flow on branch $e \, {\in} \, \edges$ under contingency $c \, {\in} \, \contingencies$ are denoted by $\vbranch^{c}_{e}$ and $\pf^{c}_{e}$, respectively.
    Branches' post-contingency status are given by
    \begin{align}
        \label{eq:contingency:vbranch}
        \vbranch^{c} &= \text{Diag}(w^{c}) \times \vbranch^{\emptyset}
            && \forall c \, {\in} \, \contingencies,
    \end{align}
    where $w^{c} \, {\in} \, \{0, 1\}^{|\edges|}$ and $w^{c}_{e} \, {=} \, 0$ if $e \, {\in} \, c$ and $w^{c}_{e} \, {=} \, 1$ otherwise.
    Post-contingency power flows are computed using Ohm's law and Kirchhoff's current law
    \begin{align}
        \label{eq:contingency:ohm}
        \pf_{e}^{c} &= b_{e} \vbranch_{e}^{c} (\va_{\dst(e)}^{c} - \va_{\org(e)}^{c}),
            && \forall e, c \, {\in} \, \edges \times \contingencies,
        \\
        \label{eq:contingency:kcl}
        \pg_{i}^{c} + \sum_{e \, {\in} \, \edges^{-}_{i}} \pf_{e}^{c}  &= \pd_{i}^{c} + \sum_{e \, {\in} \, \edges^{+}_{i}} \pf_{e}^{c}
        && \forall i, c \, {\in} \, \buses \times \contingencies,
    \end{align}
    where $\va^{c}_{i}$ is the phase angle at bus $i \, {\in} \, \buses$ under contingency $c \, {\in} \, \contingencies$.
    Thermal limits read
    \begin{align}
        \label{eq:contingency:thermal}
        | \pf_{e}^{c}| & \leq \pfmax_{e},
            && \forall e, c \, {\in} \, \edges \times \contingencies.
    \end{align}

    An important contribution of the paper is a new formulation for representing the post-contingency energized/de-energized status of each bus using constraints
    \begin{align}
        \label{eq:contingency:reference}
            \vpi_{r}^{c} & =1,
                && \forall c \, {\in} \, \contingencies,\\
        \label{eq:contingency:connected}
            |\vpi_{\org(e)} - \vpi_{\dst(e)}| &\leq 1 - \vbranch_{e}^{c}, 
                && \forall e, c \, {\in} \, \edges \times \contingencies,\\
        \label{eq:contingency:disconnected}
            \vpi_{i}^{c} & \leq \sum_{e \, {\in} \, \gc} \vbranch^{c}_{e}, 
                && \forall i, c, \gc \, {\in} \, \buses \times \contingencies \times \GC_{i}.
    \end{align}
    As proven in Theorem~\ref{thm:pi}, constraints \eqref{eq:contingency:reference}-\eqref{eq:contingency:disconnected} ensure that $\vpi^{c}_{i}$ is set to the correct value.

    \begin{theorem}
        \label{thm:pi}
        Let $c \subseteq \edges$, let $\vbranch \, {\in} \, \{0, 1\}^{\edges}$ and let $\vpi^{c} \, {\in} \, \{0, 1\}^{|\buses|}$ satisfy constraints \eqref{eq:contingency:reference}-\eqref{eq:contingency:disconnected}.
        Then $\forall i, \vpi_{i}^{c} = 1 \Leftrightarrow i \, {\in} \, \buses^{c}$.
    \end{theorem}
    \begin{proof}
        Considering $\Leftarrow$, if $i \, {\in} \, \buses^{c}$, there exists a connected path from the reference to $i$.
        Since $\vpi^{c}_{r} = 1$, every bus on that path is energized by recursion using \eqref{eq:contingency:connected}.
        Furthermore, since $i \, {\in} \, \buses^{c}$, any cutset $\gc \, {\in} \, \GC_{i}$ must have at least one active branch, so that the right hand side of \eqref{eq:contingency:disconnected} does not constrain $\vpi_i^c$.\\
        We prove $\Rightarrow$ by negation. We assume $i \notin \buses^{c}$, implying that the set of open branches $\mathcal{B}=\{e\in\mathcal{E} | \vbranch^c_e=0 \}$ is sufficient to separate bus $i$ from the reference. This means there exists a (sub)set $\kappa \subseteq \mathcal{B}$ that is a cutset, i.e.~$\kappa \in \GC_i$, and \eqref{eq:contingency:disconnected} evaluated for $\kappa$ constrains $\vpi_i^c$ to 0. Moreover, for open branches ($\vbranch^c_e=0$), \eqref{eq:contingency:connected} does not constrain $\vpi_i^c$.
    \end{proof}

    An important consequence of Theorem \ref{thm:pi} is that binary requirements on $\vpi$ can be relaxed as they are naturally enforced by \eqref{eq:contingency:reference}-\eqref{eq:contingency:disconnected}, as proven by Theorem \ref{thm:pi_continu}, yielding
    \begin{align}
        \label{eq:contingency:pi}
        \pi^{c}_{i} \in [0, 1], && \forall c, i \in \contingencies \times \buses.
    \end{align}
    This result is a fundamental contribution of the paper.

    \begin{theorem}
        \label{thm:pi_continu}
        Let $c \, {\subseteq} \, \edges$, $\vbranch \, {\in} \, \{0, 1\}^{\edges}$ and let $\vpi^{c} \, {\in} \, [0, 1]^{|\buses|}$ satisfy constraints \eqref{eq:contingency:reference}-\eqref{eq:contingency:disconnected}.
        Then $\vpi^{c} \in \{0,1\}^{|\buses|}$.
    \end{theorem}
    \begin{proof}
        Immediate from the proof of Theorem \ref{thm:pi}.
    \end{proof}

\subsection{MIP formulation of OTSD}
\label{sec:formulation:MILP}

    The complete MIP formulation of OTSD is presented below
    \begin{subequations}
    \label{eq:completeMIP}
    \begin{align}
        \min_{\vbranch, \va, \pf, \cf, \vpi, \boldsymbol{\sigma}} \quad &
            \sum_{c \in \contingencies} p^{c} \LL^{c}\\
            \text{s.t.} \quad
            & \eqref{eq:basecase:vbranch} - \eqref{eq:basecase:shadow_bigM}\\
            & \eqref{eq:contingency:demand} - \eqref{eq:contingency:pi}
    \end{align}
    \end{subequations}
    Although several constraints include non-linear terms like absolute values or bilinear products, these can be converted to MILP form using standard reformulations as follows.
    Absolute value terms of the form $|x| \, {\leq} \, y$, which appear in constraints \eqref{eq:basecase:thermal}, \eqref{eq:contingency:thermal} and \eqref{eq:contingency:connected},  are reformulated as $-y \, {\leq} \,  x \, {\leq} \, y$.
    Next, constraints \eqref{eq:basecase:ohm}, \eqref{eq:contingency:generation} and \eqref{eq:contingency:ohm} involve bilinear products of the form $y \, {=} \, x \, {\times} \, z$ where $z \, {\in} \, \{0, 1\}$ and $x \, {\in} \, [l, u]$ is continuous.
    Such bilinear products are reformulated in MILP form as
    \begin{subequations}
    \label{eq:bilinear:bigM}
    \begin{align}
        \label{eq:bilinear:bigM:1}
        lz \leq y & \leq u z, \\
        \label{eq:bilinear:bigM:2}
        (1-z) l \leq x - y &\leq (1-z) u.
    \end{align}
    \end{subequations}
    Note that the reformulation \eqref{eq:bilinear:bigM} remains valid for the bilinear product $\boldsymbol{\sigma}{\times}\vpi$ in \eqref{eq:contingency:generation}, because $\vpi$ is binary in any feasible solution per Theorem \ref{thm:pi_continu}.

    Formulation \eqref{eq:completeMIP} uses fewer variables than the formulation proposed in \cite{jeanson2025riskbasedapproachoptimaltransmission} by removing post-contingency virtual flows, but introduces an exponential number of constraints via \eqref{eq:contingency:disconnected}.
    Nevertheless, noting that constraints \eqref{eq:contingency:disconnected} can be separated efficiently, e.g., using a graph connectivity test, they can be added lazily in a callback.
    Initial experiments on small systems revealed that this callback is rarely needed, if at all, thus suggesting that constraints \eqref{eq:contingency:disconnected} may be redundant in practice.
    This empirical finding is validated by the analysis below.

    Consider a solution $(\vbranch, \va, \pf, \pf^{\star}, \vpi, \boldsymbol{\sigma})$ satisfying all constraints in \eqref{eq:completeMIP} except \eqref{eq:contingency:disconnected}, and let $c \,{\in} \, \contingencies$.
    First note that, if the grid remains connected following $c$, then constraints \eqref{eq:contingency:reference} and \eqref{eq:contingency:connected} ensure that $\vpi^{c}_{i} = 1, \forall i \in \buses$, and no constraint \eqref{eq:contingency:disconnected} is active for that contingency.
    Next, assume that $c$ leads to a loss of connectivity.
    Noting that i) $\vpi^{c}_{i} \, {=} \, 1, \forall i \, {\in} \, \buses^{c}$ using the above argument, and that ii) $\pf^{c}_{e} \, {=} \, 0, \forall e \, {\in} \, [\buses^{c}, \bar{\buses}^{c}]$, it follows that
    $\boldsymbol{\sigma}^{c} = \frac{\sum_{i \in \buses^{c}} \pdref_{i}}{\sum_{i \in \buses^{c}} \pgref_{i}}$ to ensure power balance in the energized area.
    In addition, power balance in the de-energized area reads
    \begin{align}
        \label{eq:contingency:balance}
        \sum_{i \in \bar{\buses}^{c}} \vpi^{c}_{i} (\boldsymbol{\sigma}^{c} \pgref_{i} - \pdref_{i}) &= 0.
    \end{align}
    Therefore, if the de-energized area contains only loads or only generators, then \eqref{eq:contingency:balance} can only hold if $\vpi^{c}_{i} \,{=} \, 0, \forall i \in \bar{\buses}^{c}$.
    If, on the other hand, the de-energized area contains both generators and loads, then a non-identically-zero $\vpi$ solution is possible only if a subset of generations and loads in the de-energized area are perfectly balanced, because \eqref{eq:contingency:connected} enforces the value of $\vpi^c_i$ to be equal for the buses within that (internally connected) area.
    This situation is extremely unlikely in practice, which explains why constraints \eqref{eq:contingency:disconnected} are redundant almost all the time. In practice, one can verify the validity of results after optimization, or implement a call-back. This analysis sheds some light on the fact that  makes formulation \eqref{eq:completeMIP} a more robust foundation for scaling to large-scale systems compared to the original OTSD formulation proposed in \cite{jeanson2025riskbasedapproachoptimaltransmission}.
    

%% file: completeHeuristic.tex
\section{Solution Methodology}
\label{sec:heuristic}

    Preliminary experiments revealed that commercial MIP solvers struggle to find feasible solutions when presented with the extensive formulation \eqref{eq:completeMIP}.
    In an attempt to address these limitations, we implemented a column-and-constraint algorithm where \eqref{eq:completeMIP} is initialized with a (possibly empty) subset of contingencies, and additional contingencies are added iteratively, leading to the addition of relevant variables and constraints \eqref{eq:contingency:demand}-\eqref{eq:contingency:pi}.
    We also implemented a Benders decomposition scheme, made possible by leveraging the result of Theorem \ref{thm:pi_continu}.
    Nevertheless, neither of these strategies resulted in acceptable performance.
    Therefore, this section presents a heuristic algorithm to produce high-quality solutions in short computing times.
    We conjecture that this heuristic, combined with the proposed MILP formulation, provide a strong foundation for algorithmic approaches capable of solving OTSD exactly on sub-transmission grids with hundreds of buses.

\begin{figure}[!t]
    \centering
    \includegraphics[width=.9\linewidth]{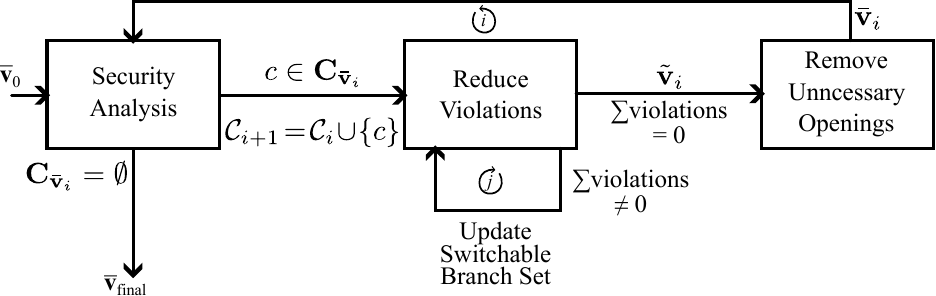}
    \caption{\textbf{Fast feasible solution finder heuristic.}
    The algorithm terminates when the \emph{security analysis} with configuration $\vsol_i$ finds no contingency causing violation.
    Otherwise, the main loop continues, adding the most constraining contingency to the subset $\contingencies_i$. 
    \emph{Reduce violations} step minimizes the violation; if zero is reached, $\vfsol_i$ is feasible, otherwise, the inner loop indexed by $j$ expands the set of switchable branches.
    Finally, \emph{unnecessary openings are removed}, yielding the resulting configuration $\vsol_i$.}
    \label{fig:algorithm}
    \vspace{-1em}
\end{figure}

\subsection{Skeleton of the heuristic}

    \begin{algorithm}
    \small
    \newcommand{\valg}{\textbf{v}}
    \newcommand{\newc}{c_{new}}
    \newcommand{\C}{\mathbf{C}}
    \newcommand{\toInc}{toExpand}
    \algdef{SE}[DOWHILE]{Repeat}{Until}{\algorithmicrepeat}[1]{\algorithmicuntil\ #1}
    \algtext*{EndRepeat} 
    
    \caption{Fast feasible solution finder}
    \label{alg:FastFeasibleSolutionFinder}
    \begin{algorithmic}[1]
        \State $\vsol \gets [1, \dots,1]$
        \State $\C \gets \textbf{SecurityAnalysis}(\vsol)$
        \For{$c \in \C$}
            \State $\MB[c] \gets \VBf{\vsol}{c}$
                \State $\nbhop[mb] \gets \nbhopO$ \textbf{for} {$mb \in \MB[c]$}
        \EndFor
        \While{True}
            \Repeat
                \For{$c \in \C, mb \in \MB[c]$}
                    \State $\PBO \gets \PBO \cup(\textbf{\hop}(mb, \nbhop[mb]))$
                \EndFor
                \State $(\vfsol, \crv) \gets \ReduceViol(\contingencies, \MB, \PBO)$
                \State $\toInc \gets \text{Set}()$
                \For{$c \in \crv$}
                    \State $\MB[c] \gets \MB[c] \cup \{vb\}$ \textbf{for} $vb \in \VBf{\vfsol}{c}$
                    \State $\toInc \gets \toInc \cup \{vb\}$ \textbf{for} $vb \in \MB[c]$
                \EndFor
                \For{$vb \in \toInc$}
                    \If{$vb \notin \text{Keys}(\nbhop)$}
                        \State $\nbhop[vb] \gets \nbhopO$
                    \Else
                        \If{$\nbhop[vb] = \nbhopmax$}
                            \State Return INFEASIBLE
                        \EndIf
                        \State $\nbhop[vb] \gets \nbhop[vb] +1$
                    \EndIf
                \EndFor
            \Until{$\crv \ne \emptyset$}
            \State $\vsol \gets \RemoveUO(\valg_{\FS}, \C)$
            \State $\C2\gets \textbf{SecurityAnalysis}(\vsol)$
            \If{$\C2 = \emptyset$}
                \State Return $\vsol$
            \EndIf
            \State $\newc \gets \MostConstraining(\C2)$
            \State $\C \gets \C \cup \{\newc\} $
            \State $\MB[c] \gets \VBf{\vsol}{c}$
            \For{$mb \in \MB[c]$}
                \If{$mb \notin \text{keys}(\nbhop)$}
                    \State $\nbhop[mb] \gets \nbhopO$
                \EndIf
            \EndFor
        \EndWhile
        \end{algorithmic}

    \end{algorithm}
    
    The algorithm is sketched in Figure \ref{fig:algorithm} and detailed in Algorithm \ref{alg:FastFeasibleSolutionFinder}.
    The goal is to find a feasible solution quickly by targeting key contingencies and limiting the branches considered for switching.
    
    Consider the function $\SA(\vsol) \rightarrow \contingenciesSa{\vsol}$ that runs over the whole $\contingencies$  a security analysis of the grid with the switching \emph{configuration} $\vsol$ and returns the set of \emph{contingencies} $\contingenciesSa{\vsol}$ that lead to violations of the thermal limits.

    Denote by $\vsol_0$ the initial configuration in which all branches are closed and let $\contingencies_1=\contingenciesSa{\vsol_0}$. At any point, $\VBf{\vsol}{c}$ denotes the set of branches subject to a violation under configuration $\vsol$ and contingency $c$.

    At  iteration $i$, the goal is to find a feasible solution $\vsol_i$ that satisfies -- without unnecessary disconnection  -- the constraints \eqref{eq:basecase:vbranch}-\eqref{eq:contingency:connected} on the subset $\contingencies_i$.
    If no solution is found, the problem is infeasible for $\contingencies_i$, a fortiori for $\contingencies$, and \emph{the program \eqref{eq:completeMIP} is also infeasible}.
    Conversely, if $\contingenciesSa{\vsol_i} = \emptyset$, $\vsol_i$ is a solution for all $\contingencies$ and \emph{$\vsol_i$ is a solution for \eqref{eq:completeMIP}}.

    Then, we iterate by picking up $c_{i+1} \in \contingenciesSa{\vsol_i}$
    which is returned by the function
    $
        \MostConstraining(\contingencies) 
        \rightarrow \text{argmax}_{c \in \contingenciesSa{\vsol_i}} 
        |\VBf{\vsol_i}{c}|
    $
    and define $\contingencies_{i+1}=\contingencies_i\cup\{c_{i+1}\}$.
    Note that as $\vsol_i$ is a solution on $\contingencies_i$, $\contingencies_i\cap\contingenciesSa{\vsol_i}=\emptyset$ and $c_{i+1}\notin\contingencies_i$.

    Therefore each iteration of the \emph{main loop} counts three steps detailed in the following subsections:
    perform a security analysis (\ref{sec:security_analysis});
    find a feasible solution $\vfsol_i$ (\ref{sec:findFeasibleSolution});
    simplify it to $\vsol_i$ by removing unnecessary openings (\ref{sec:removeUnnecessaryOpenings}).

\subsection{The security analysis}
\label{sec:security_analysis}

    \mt{
        Note that checking for N-1 feasibility is fairly fast: it can be done by rebalancing + DC power flow computation. Noting that, in real life, there is usually a small number of binding constraints, we first propose an iterative algorithm where we iteratively add contingencies to the MIP: solve MIP, security analysis, add contingency, solve MIP, security analysis, etc...\\
        An important thing to note is that, at each iteration, we are add constraints and variables (unlike, e.g., PTDF or LODF formulations that only add constraints), because we need to account for potential loss of connectivity under contingency\\
        Note that preliminary experiments show that the MIP becomes intractable after even 6 contingencies on 118 system...
        ... which motivates v2 that focuses on finding feasible solutions fast.}

    The formulation \eqref{eq:completeMIP} can be used as a security-analysis program consistent with the OTS model, identifying de-energized areas and rebalancing the system accordingly.
    It is obtained by omitting \eqref{eq:contingency:thermal} and fixing the binary variables $\vbranch$ to a given configuration.
    With no real degrees of freedom or infeasibility constraints, it solves pretty rapidly and establishes the system flows for both the ‘N’ and ‘N-1’ cases.
    Using this solver-based approach, a complete security analysis over $\contingencies$ was achieved in 12 s on the \textbf{118-IEEE} system.

    Performance was further improved replacing that solver-based implementation with a simplified procedure that reduces each contingency case to an independent linear system solvable within milliseconds.
    It first identifies bridges -- individual closed branches whose disconnection would split the grid into two islands -- corresponding to potential de-energization events.
    The main connected component is then defined and rebalanced, and the resulting power flows are computed using Power Transfer Distribution Factors (PTDF).
    This improvement reduced computation time to under 10 ms for the same system.

\subsection{Find a feasible solution $\vfsol_i$}
\label{sec:findFeasibleSolution}

    Examining the solver’s behavior when applied to the implementation of \eqref{eq:completeMIP} revealed a seemingly random exploration, prompting a reformulation of the problem.
    Instead of enforcing thermal limits, they were relaxed and penalized in the objective, allowing the solver to focus on decisions that effectively reduce violations, which in turn improves convergence speed.
    This is done using the program \eqref{eq:reduceViolations}. 
    \begin{subequations}
    \label{eq:reduceViolations}
    \begin{align}
        \min_{\vbranch, \va, \pf, \vpi, \sigma, \ol} \quad &\sum_{(c,e) \in \contingencies_i \times \edges} \ol^c_e\\
        \text{s.t.} \quad
        & \label{eq:TNR:thermalOverload}
        |\pf^{c}_{e}| \leq \pfmax_{e} + \ol^c_e 
            && \forall c, e \in \contingencies_i \times \edges\\
        & \label{eq:ThermalOverloadPositive}
            ol_c^e \ge 0
            && \forall c,e \in \contingencies_i \times \edges\\
        & \label{subeq:removeThermalLimitations}
            \eqref{eq:basecase:vbranch}{-}\eqref{eq:basecase:kcl},
            \eqref{eq:basecase:shadow_flow}{-}
            \eqref{eq:contingency:connected},
            \eqref{eq:contingency:pi}
            && \forall c \in \contingencies_i
    \end{align}
    \end{subequations}
    For branch $e$ and contingency $c$, a non-negative slack $ol_e^c$ relaxes the flow limit  \eqref{eq:TNR:thermalOverload} and is penalized in the objective.
    If the objective is zero, the solution is feasible for \eqref{eq:completeMIP} on $\contingencies_i$.

    Despite this relaxation, the problem remains challenging even on a restricted subset $\contingencies_i$.
    To limit the number of binary decision variables, the algorithm starts from a reduced instance that is gradually expanded, introducing degrees of freedom only where needed.
    The remaining violations then indicate where issues persist and around which these degrees of freedom should expand.
    Each solution serves as a warm start for the next iteration until a zero-objective solution is reached; this \emph{inner loop} is indexed by $j$.
    
    Let $\MB[c]$ denote the set of \emph{monitored branches} associated with contingency $c$, and $\MB_i^j[c]$ its state at iteration $ij$.
    For a given $c$, the monitored branches are those branches whose limits have been violated in any previous security analysis of the main loop, or in any previous call to \ReduceViol\ where $c$ was involved.
    Let $\MB_i^j=\bigcup_{c \in \contingencies_i}\MB_i^j[c]$ be the set of monitored branch for iteration $ij$.
    
    Denote $\PBO$ the \emph{Switchable Branch Set}, and $\PBO_i^j$ its state at iteration $ij$.
    $\PBO$ is the active set of degrees of freedom for this iteration.
    The approach relies on the fact that opening a branch primarily affects nearby flows; hence, the branches most likely to relieve a violation are those closest to it.
    In \eqref{eq:PBOrestriction}, each branch that is \emph{not in} $\PBO$ is forced closed except in its own tripping case.
    
    \begin{equation}\label{eq:PBOrestriction}
        \vbranch_e^c = 1 \quad \quad \quad \quad \forall c \in \contingencies_i,\ \forall e\notin \PBO\setminus{c}
    \end{equation}
    
    The function $\ReduceViol(\contingencies_{i}, \PBO)$, which returns $(\vfsol, \crv)$, corresponds to the program \eqref{eq:reduceViolations} with the additional constraints \eqref{eq:PBOrestriction}.
    The configuration $\vfsol$ minimizes flow violations, while $\crv\subset\contingencies_i$ denotes contingencies whose violations could not be reduced to $0$.

    \ReduceViol\ runs with $\PBO$ iteratively expanded around the monitored branches $\MB_i^j$ according to \eqref{eq:PBO}. 
    \begin{equation} \label{eq:PBO}
        \PBO_i^j = \bigcup_{e \in \MB_i^j} \hop(e, {\nbhop}_i^j[e])
    \end{equation}
    $\hop(e,l)$ is a function defined on $\edges\times\mathbb{N}$  that returns the set of all edges that are connected to the edge $e$ through a connected travel of length $l$. It follows that $\hop(e,0)=\{e\}$ and $\hop(e,l+1)\subseteq \hop(e,l)$ (Fig.\ref{fig:Hop}).

    \begin{figure}[!ht]
        \centering
        \includegraphics[width=.6\linewidth]{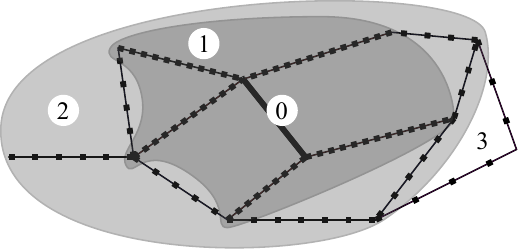}
        \caption{The $\hop(e,l)$ returns a set of branches composed of the branches separated from the origin branch $e$ by paths of length lower or equal to $l$.}
        \label{fig:Hop}
        \vspace*{-1em}
    \end{figure}

\input{resultTable}

    For each branch $e \in \MB$, let $\nbhop[e]$ denote the \emph{number of hops} to be considered, and ${\nbhop}_i^j[e]$ its value in iteration $ij$.
    When the flow in branch $e$ violates its limit in \SA\ or \ReduceViol\ for the first time, it is added to $\MB$ and $\nbhop[e]=\nbhopO$.
    If $e$ is already in $\MB$, $\nbhop[e]$ is incremented whenever there exists a contingency $c\in\contingencies_i$ such that $e\in\MB_i^j[c]$ and $c$ still produces residual violations in \ReduceViol, even if $e$ itself is not violated.
    Consequently, the degrees of freedom remain localized around the monitored branches and expand only around those newly violated or associated with contingencies that continue to cause issues.
    $\PBO_i^j$ is updated until \ReduceViol\ returns a solution $(\vfsol_i^j, \emptyset)$, implying that configuration $\vfsol_i=\vfsol_i^j$ satisfies all thermal limits for the contingencies in $\allcases_i$.  
    The case is deemed infeasible whenever resolving remaining violations would require increasing some $\nbhop[e]$ beyond $\nbhopmax$.

\subsection{Remove unnecessary openings}
\label{sec:removeUnnecessaryOpenings}
    Multiple feasible solutions may satisfy the constraints of the previous problem, and some -- including $\vfsol_i$ -- may contain unnecessary branch openings.
    Since opening branches generally weakens the grid, such redundant openings can trigger additional violations when evaluating new contingencies beyond $\contingencies_i$.
    Therefore, before performing the security analysis, these superfluous openings shall be removed.
    The program \eqref{eq:Simplify} fulfills this role by generating a simplified solution that meets the same constraints as $\vfsol_i$ while minimizing the number of open branches.
    \begin{subequations}
        \label{eq:Simplify}
        \begin{align}
            \min_{\vbranch, \va, \pf, \vpi, \sigma} \quad & \sum_{e \in \edges} 1 - \vbranch^\emptyset_e\\
            \text{s.t.} \quad
        & \eqref{eq:basecase:vbranch}-            \eqref{eq:contingency:connected},
            \eqref{eq:contingency:pi}
        &&\forall c \in \contingencies_i \\
        & \vbranch_e^c \ge \vfsol_e^c &&\forall c \in \contingencies_i,\ \forall e\in \edges
        \end{align}
    \end{subequations}

    $\RemoveUO$ applies this program to $(\vfsol_i, \contingencies_i)$ and returns $\vsol_i$.
    It runs efficiently, as $\vfsol_i$ contains only a few openings.

%% file: resultTable.tex

\begin{table*}[!t]
\centering
\caption{Computation Results}
\label{tab:results}
\resizebox{\textwidth}{!}{
\begin{tabular}{lrc@{\hskip 1.5em}cccccc@{\hskip 1.5em}cccccc@{\hskip 1.5em}cccr}
    \toprule
        &&
        & \multicolumn{6}{c}{\originAlg}
        & \multicolumn{6}{c}{\completeAlg}
        & \multicolumn{3}{c}{Heuristic}
        & 
        \\
    \cmidrule(r){4-9}
    \cmidrule(r){10-15}
    \cmidrule(r){16-19}
       Case
        & TLF 
        & SR
        & \textbf{$T_{1}$} & Obj$_{1}$ & $O_{1}$ 
        & \textbf{$T_{f}$} & Obj$_{f}$ & $O_{f}$ 
        & \textbf{$T_{1}$} & Obj$_{1}$ & $O_{1}$ 
        & \textbf{$T_{f}$} & Obj$_{f}$ & $O_{f}$ 
        & \textbf{$T_{1}$} & Obj$_{1}$ & $O_{1}$ 
        & Speedup
    \\
    \midrule
    14\_IEEE    & 100\% & 0         & -        & -         & -  & 371~ms & \bf{2.37} & 2      & -       & -         & 2  & 366~ms  & \bf{2.37} & 2     & 15~ms   & \bf{2.37} & \bf{1}  & 25   \\ 
    24\_IEEE    & 100\% & 0         & -        & -         & -  & 1.4~s  & \bf{1.66} & \bf{2} & -       & -         & -  & 6~s     & \bf{1.66} & 4      & 40~ms   & \bf{1.66} & \bf{2} & 35   \\ 
    30\_IEEE    & 120\% & .035      & 2~s      & 9.07      & 10 & 19~s   & \bf{6.82} & \bf{4} & 8~s     & 7.97      & 10 & 36~s    & \bf{6.82} & \bf{4} & 208~ms  & \bf{6.82} & \bf{4} & 10   \\ 
    30\_IEEE    & 100\% & .035      &  -       & -         & -  & 8.7~s  & inf       & -      & -       & -         & -  & 9.2~s   & inf       & -      & 1.9~s   &  inf      & -      & 5    \\ 
    57\_IEEE    & 200\% & \bf{.038} & -        & -         &    & 17~s   & \bf{.038} & \bf{0} & -       & -         & -  & 19~s    & \bf{.038} & 10     & 7~ms    & \bf{.038} & \bf{0} & 2,400 \\ 
    57\_IEEE    & 150\% & \bf{.038} & -        & -         &    & 45~s   & \bf{.038} & 2      &         &           &    & 40~s    & \bf{.038} & 8      & 91~ms   & \bf{.038} & \bf{1} & 495  \\ 
    57\_IEEE    & 120\% & .038      & 8~min    & 14.9      & 14 & 63~min & \bf{6.37} & 10     & 117~s   & 12.5      & 12 & $>$ 3~h & -         & -      & 138~ms  & 8.6       & 5      & 3,480 \\ 
    57\_IEEE    & 100\% & .038      & 150~s    & 27.4      & 18 & 69~min & \bf{7.33} & \bf{6} & 290~s   & 11.0      & 9  & $>$ 3~h & -         & -      & 533~ms  & 11.6      & 8      & 280  \\ 
    73\_RTS     & 100\% & 0         & $>$ 3~h  & -         &    & -      & -         &        & $>$ 3~h & -         &    & -       & -         &        & 66~s    & .83       & 5      & $\infty$\\ 
    118\_IEEE   & 150\% & 2.99      & $>$ 3~h  & -         &    & -      & -         &        & $>$ 3~h & -         &    & -       & -         &        & 80~s    & 7.9       & 10     & $\infty$\\ 
    118\_IEEE   & 125\% & 2.99      & $>$ 3~h  & -         &    & -      & -         &        & $>$ 3~h & -         &    & -       & -         &        & $>$ 3~h & -         &        & $\infty$\\ 
    200\_activ  & 100\% & \bf{17.4} & -        & -         & -  & 10~h   & \bf{17.4} & 13     & $>$ 3~h & -         &    & -       & -         &        & 49~ms   & \bf{17.4} & \bf{0} & 735,000\\ 
    200\_activ  &  60\% & 17.4      & -        & -         &    & -      & -         &        & -       & -         &    & -       & -         &        & 1~s     & 18.7      & 2      & - \\ 
    200\_activ  &  55\% & 17.4      & -        & -         &    & -      & -         &        & -       & -         &    & -       & -         &        & 140~ms  & bc\_inf   &        & - \\ 
    300\_IEEE   & 300\% & 51        & -        & -         &    & -      & -         &        & -       & -         &    & -       & -         &        & 23~s    & 78        & 4      & - \\ 
    300\_IEEE   & 200\% & 51        & -        & -         &    & -      & -         &        & -       & -         &    & -       & -         &        & $>$ 3~h & -         &        & - \\ 
    \bottomrule
\end{tabular}}
\begin{tablenotes}
\footnotesize
    \item 
    TLF: Thermal Limit Factor; 
    SR: Structural Risk (bold when known optimal);
    \originAlg\ is the algorithm from \cite{jeanson2025riskbasedapproachoptimaltransmission}.
    \completeAlg\ is the full implementation of the program of Section \ref{sec:problem}
    and Heuristic is the heuristic of Section \ref{sec:heuristic}.
    $T_x$, $Obj_x$, $O_x$ respectively corresponds to runtime, objective value, and number of openings of the solution.
    Index $x\in\{1,f\}$ refers to the first and final feasible solutions; the first is omitted when identical to the final.
    Optimal values are in bold (only \originAlg\ and \completeAlg\ qualifies optimality).
    When the solution is optimal, the smallest \textbf{$O_x$} found is in bold.
    inf: infeasible, bc\_inf: base case infeasible.
    Speedup is defined as the ratio of $T1$ from \originAlg\ to that of the heuristic.
\end{tablenotes}
\vspace{-1.5em}
\end{table*}

%% file: results.tex
\section{Results and discussion}
\label{sec:results}

We evaluated three algorithms on various grids: the \emph{\originAlg} algorithm from \cite{jeanson2025riskbasedapproachoptimaltransmission}, the \emph{\completeAlg} algorithm -- a direct implementation of Section~\ref{sec:problem} using lazy constraints for \eqref{eq:contingency:disconnected} -- and the proposed \emph{Heuristic} from Section~\ref{sec:heuristic}. 
All experiments were run in \emph{Julia} with \emph{Gurobi~12} on a MacBook Air M3. 
Network data and thermal limits were taken from \emph{PGLib}. The contingency set comprises all branches, with equal probability assigned to each.
To avoid trivial or infeasible cases, thermal limits were scaled by a \emph{Thermal Limit Factor} (TLF), which under the DC approximation yields constraints equivalent to scaling injections inversely to the TLF. 
For some networks, Multiple TLF values were tested to assess the effect of network stress.
The \emph{structural risk} (SR) corresponds to the objective value when all branches are closed and thermal limits are not enforced.
It provides a lower bound for the problem: when nonzero, it represents the amount of load that cannot be secured because it depends on a single branch.

Table~\ref{tab:results} summarizes the results. 
Because our heuristic only focuses on finding a feasible solution, we also report the times and objective values at which the \originAlg\ and \completeAlg\ algorithms first reach  feasibility; their final solutions are optimal.
When the computational time exceeded 3 hours without producing a result, a generic ``$> 3~h$'' was recorded, with the exception being the case \textbf{200\_activ} for which a solution was found beyond this time and the time was retained.
A ``–'' indicates simulations not attempted, or in the first-result columns, cases where the first and final results coincide.

The \originAlg\ and \completeAlg\ algorithms perform similarly, with a slight advantage for \originAlg, while our heuristic consistently outperforms both in reaching feasibility.
In the cases studied where the optimal value is known, we observe that it is often attained by the heuristic.
For the largest cases, which were intractable for the \originAlg\ and \completeAlg\ algorithms, the heuristic was able to find solutions in a limited time.
This suggests that limiting the search to the neighborhood of congested areas is sufficient to capture near-optimal, and frequently optimal, solutions. 
In several instances, the first feasible solution from the exact methods was already close to optimal, confirming the relevance of feasibility-oriented heuristics.

When no line opening is needed, the solution is trivially optimal (case \textbf{57\_IEEE TLF~200\%} and \textbf{200\_active TLF~100\%}). The heuristic reaches this solution almost instantly as it only requires a security analysis.
In contrast, the other algorithms fail to recognize the triviality of the case and take longer to conclude.
Conversely, when no feasible solution exists, the heuristic runtime becomes comparable to that of the other algorithms, since proving infeasibility would require an exploration of the entire search space, something achievable only by setting $\nbhopmax$ up to the diameter of the line graph.
However, keeping the solution within the vicinity of the constraints makes sense for the operator, which justifies a limited $\nbhopmax$.
Indeed, when a solution requires widespread topology changes beyond the vicinity of the violated branch, it likely reflects a fragile operating state, close to security limits, and highly vulnerable to disturbances.
In such cases, the operator would typically avoid relying on such configuration and instead opt for more costly but secure preventive actions.

For case \textbf{57\_IEEE TLF~150\%}, the system is constrained and requires branch openings to withstand contingencies. The resulting risk matches the structural risk, indicating that the OTSD’s de-energization feature remains unused and the outcome would coincide with a classical OTS.
By contrast, most cases show that the OTSD enables efficient solutions that would otherwise require costly actions, thereby extending the feasible operating domain at the cost of a controlled increase in risk exposure.

In summary, the heuristic delivers high-quality solutions at a fraction of computational cost. 
It effectively captures the structure of optimal solutions without the overhead of formal optimality proofs, demonstrating a practical trade-off between speed and exactness.

%% file: conclusion.tex
\section{Conclusion}
\label{sec:conclusion}

This paper considers Optimal Transmission Switching with De-energization, a challenging problem that reflects real-life operating practices yet has received little attention in the literature.
The paper improves the state of the art \cite{jeanson2025riskbasedapproachoptimaltransmission} by i) proposing a new MILP formulation that does not require binary variables for representing post-contingency loss of connectivity, and ii) developing an iterative heuristic algorithm to find feasible solutions in reasonable computing times.
These advances mark an additional step towards solving OTSD exactly on sub-transmission grids.

Computational experiments on transmission grids with up to 300 buses confirm the difficulty of solving OTSD exactly using existing MIP technology.
In particular, commercial MIP solvers like Gurobi struggle to find feasible solutions (even trivial ones).
Nevertheless, numerical results demonstrate the efficiency of the proposed heuristic which, in most cases, yields high-quality feasible solutions within seconds, a ${>}1000\times$ speedup compared to solving the extensive MIP with Gurobi.
Future work will investigate further performance and scalability improvements for larger grids, as well as the incorporation of bus splitting actions in the problem formulation.